\documentclass[letterpaper, 10 pt, conference]{ieeeconf}  
\IEEEoverridecommandlockouts                              
\usepackage{graphicx}
\usepackage{amsmath}
\usepackage{amsfonts}
\usepackage[acronym]{glossaries}
\usepackage{xcolor}
\usepackage{xurl}
\usepackage{cite}
\usepackage{multirow}
\usepackage{float}
\usepackage{algorithm,algorithmic}
\let\labelindent\relax
\usepackage{enumitem}

\newtheorem{thm}{Theorem}

\newtheorem{prop}{Proposition}

\newtheorem{lemma}{Lemma}
\newtheorem{rmk}{Remark}
\newtheorem{cor}{Corollary}

\def\arxiv{1} 

\title{\LARGE \bf
Soft projections for robust data-driven control
}

\author{Andr\'as Sasfi, Jaap Eising, Florian D\"orfler%
\thanks{A. Sasfi and F. D\"orfler are with the Department of Information Technology and Electrical Engineering, ETH Z\"urich, Switzerland, {\tt\footnotesize \{asasfi, doerfler\}@control.ee.ethz.ch}. J.~Eising is with ENTEG, University of Groningen, the Netherlands, {\tt\footnotesize j.eising@rug.nl}. This work was supported by the SNF/FW Weave Project 200021E\_20397}}

\newacronym{LTI}{LTI}{linear time-invariant}
\newacronym{DeePC}{DeePC}{data-enabled predictive control}

\begin{document}

\maketitle
\thispagestyle{empty}
\pagestyle{empty}

\begin{abstract}
We consider data-driven predictive control based on behavioral systems theory.
In the linear setting this means that a system is described as a subspace of trajectories, and predictive control can be formulated using a projection onto the intersection of this behavior and a constraint set.
Instead of learning the model, or subspace, we focus on determining this projection from data.
Motivated by the use of regularization in \gls{DeePC}, we introduce the use of soft projections, which approximate the true projector onto the behavior from noisy data.
In the simplest case, these are equivalent to known regularized \gls{DeePC} schemes, but they exhibit a number of benefits.
First, we provide a bound on the approximation error consisting of a bias and a variance term that can be traded-off by the regularization weight.
The derived bound is independent of the true system order, highlighting the benefit of soft projections compared to low-dimensional subspace estimates.
Moreover, soft projections allow for intuitive generalizations, one of which we show has superior performance on a case study.
Finally, we provide update formulas for soft projectors enabling the efficient adaptation of the proposed data-driven control methods in the case of streaming data.
\end{abstract}

\section{Introduction}
Direct data-driven predictive control methods~\cite{coulson2019data,berberich2020data,breschi2023data,verhoek2021data,depersis2019formulas,vanwaarde2020informativity} have recently gained significant attention in the control community. 
These formulations are fundamentally rooted in behavioral systems theory, which takes a representation-free perspective and treats dynamical systems simply as sets of trajectories~\cite{willems1997introduction}. 
For \gls{LTI} systems, this means that the system behavior can be described as a linear subspace containing all possible trajectories.
This viewpoint enables one to use raw data directly to represent the underlying system~\cite{willems2005note}.
This approach is often referred to as \textit{direct} to emphasize the contrast with traditional \textit{indirect} approaches, which require the intermediate step of first identifying a parametric model before performing controller design based on that estimate.
To ensure that these direct control formulations remain robust to measurement noise and nonlinearities, regularization is typically added to the optimization problem~\cite{dorfler2022bridging,markovsky2022data}.
As demonstrated by numerous empirical case studies, regularized data-driven control methods often perform remarkably well in practical application domains~\cite{dorfler2022bridging,markovsky2021behavioral}.

Initially, regularization was introduced into direct data-driven control formulations heuristically, lacking clear theoretical interpretability~\cite{coulson2019data}.
Subsequent research has largely focused on addressing this gap to develop more interpretable formulations~\cite{kladtke2023implicit,mattsson2024equivalence,dorfler2022bridging,breschi2023data,huang2023robust,sasfi2025gaussian}. 
While some of these studies successfully explain the implicit effects of regularization within existing frameworks~\cite{kladtke2023implicit,mattsson2024equivalence}, they do not propose novel control algorithms.
Conversely, other approaches have modified the seminal formulation~\cite{coulson2019data} to derive new, theoretically grounded regularization terms~\cite{breschi2023data,dorfler2022bridging}.
A complementary line of work~\cite{berberich2020data,berberich2022stability}, establishes closed-loop stability of regularized formulations with terminal equality constraints under bounded noise.
Beyond interpretability, a second major challenge is the efficient integration of online data updates into these control schemes, as highlighted in the surveys~\cite{berberich2024overview,markovsky2021behavioral}.
Although recent methods~\cite{sasfi2025great,jin2025online} allow for the online adaptation of behavioral subspace models, they rely on explicitly estimating the underlying subspace. 
This explicit approach has notable drawbacks: representing the behavior as a low-dimensional subspace generally yields inferior performance compared to regularized schemes~\cite{markovsky2022ddinterpolation}, and the accuracy of these estimates is highly sensitive to system order selection.

To address current challenges, we take a ground-up perspective rather than adapting the seminal formulation. 
Based on behavioral systems theory for linear systems, we view predictive control as a projection onto the behavior.
Inspired by the principle behind the direct approach to data-driven control, we learn the solution map directly instead of explicitly learning the system model or its subspace representation.
More specifically, 
we introduce the use of \textit{soft projections}~\cite{wax2021detection,wax2022vector,wax2023robust} to approximate the true projector from noisy data, which is conceptually similarly to regularization in existing methods.
Crucially, we provide rigorous bounds on the approximation error of these soft projections.
The derived error bound reveals a clear trade-off between bias and variance, which can be tuned via a parameter similar to a regularization weight.
Furthermore, this bound is completely independent of the true system order, thereby avoiding the sensitivity issues that come with explicit subspace estimates.
The proposed perspective leads to conceptually novel control formulations, offering a new alternative to approaches that focus on designing interpretable regularization terms.
Specifically, we propose two novel formulations: the first is a generalization of regularized \acrfull{DeePC}~\cite{coulson2019data}, while the second is a more interpretable formulation that exhibits superior robust performance in a case study.
Finally, we provide efficient rank-one update formulas for soft projectors, allowing the control methods to seamlessly adapt new data online.

The remainder of the paper is organized as follows.
Section~\ref{sec:prelim} contains preliminaries on behavioral systems theory, orthogonal projections, and direct data-driven control. 
Soft projectors are analyzed in detail in Section~\ref{sec:soft_proj}.
We provide novel data-driven control formulations in Section~\ref{sec:novel_methods}, and describe the online updates of soft projectors in Section~\ref{sec:online}.
Finally, Section~\ref{sec:case_study} contains a case study.

\section{Preliminaries} \label{sec:prelim}
\subsection{Behavioral systems theory} \label{sec:prelim_control}
Behavioral systems theory defines a system as a set of trajectories, called the \textit{behavior}.
In this work, we focus on finite length system trajectories constructed as $w_{t,t+L-1}=\begin{bmatrix}
    u_t^\top & u_{t+1}^\top & \dots & u_{t+L-1}^\top & y_t^\top & y_{t+1}^\top & \dots& y_{t+L-1}^\top
\end{bmatrix}^\top \in \mathbb{R}^{qL}$,
where $u_t\in\mathbb{R}^m$ and $y_t\in\mathbb{R}^{(q-m)}$ are the inputs and outputs of the system, respectively.
For \gls{LTI} systems, the \textit{restricted behavior}, $\mathcal{B}_L$ is defined as
\begin{align*}
    \mathcal{B}_L = \{w\in \mathbb{R}^{qL}~|~w \text{ is a length-$L$ trajectory of the system}\},
\end{align*}
and it is a shift-invariant subspace of the set of all possible trajectories~\cite{markovsky2022data}.
Let $n$ be the order of a minimal system realization.
Then, for \gls{LTI} systems and for large enough $L$, the restricted behavior $\mathcal{B}_L \subseteq \mathbb{R}^{qL}$ is a subspace of dimension $d = mL + n$~\cite{willems1997introduction}.
A basis for $\mathcal{B}_L$ can be constructed from a state-space representation or directly from sufficiently rich and noise free data for \gls{LTI} behaviors~\cite{markovsky2022data, willems2005note}.
The restricted behavior fully specifies the behavior in case $L$ is larger than the lag of the system~\cite{willems1997introduction}.
We focus on the restricted behavior $\mathcal{B}_L$ and refer to it simply as the behavior subsequently.

\subsection{Orthogonal projections}
Let $U\in\mathbb{R}^{qL\times d}$ be a matrix with full column rank. The orthogonal projector to the column space $\mathrm{col}(U)$ is
        $$
        P_U = U(U^\top U )^{-1}U^\top.$$
The projector to the orthogonal complement of $\mathrm{col}(U)$ is $I-P_U$.
Furthermore, $P_U = P_U^\top$ and $P_U^2 = P_U$ hold.
For two subspaces $\mathrm{col}(U_1)$ and $\mathrm{col}(U_2)$, the gap metric~\cite{bharadwaj2025robust} is
$$
d_\infty(\mathrm{col}(U_1),\mathrm{col}(U_2)) := \|P_{U_1} - P_{U_2}\|_2,
$$
and it measures the largest principal angle between the subspaces.
With other words, the gap metric quantifies the worst case error between the projection of a unit vector onto $\mathrm{col}(U_1)$ and $\mathrm{col}(U_2)$.
If the dimensions of $\mathrm{col}(U_1)$ and $\mathrm{col}(U_1)$ are different, then $\|P_{U_1} - P_{U_2}\|_2 =1$.
Note that the projector does not depend on the representation of the subspace, i.e., if $\mathrm{col}(U_1) = \mathrm{col}(U_2)$, then $P_{U_1} = P_{U_2}$.

\subsection{Behavioral predictive control}
We can use the behavioral framework to formulate a predictive control scheme as follows.
At time $t$, consider $T_\mathrm{ini}$ greater than the lag and define $w_\mathrm{ini} = w_{t-T_\mathrm{ini}-1,t-1}$.
We are interested in finding a finite input sequence, such that the future trajectory $w_\mathrm{f} = w_{t,t+T_\mathrm{f}-1}$ minimizes the control cost
\begin{align*}
    J(w_\mathrm{f}) = \sum_{i=t}^{t+T_\mathrm{f}-1} \|y_i - y_\mathrm{ref}\|_Q^2 + \|u_i - u_\mathrm{ref}\|_R^2,
\end{align*}
where $Q$ and $R$ are symmetric positive definite matrices, and $\|\cdot\|_Q$ denotes the 2-norm weighted by $Q$.
The overall trajectory $w = \begin{bmatrix}
    w_\mathrm{ini}^\top & w_\mathrm{f}^\top
\end{bmatrix}^\top\in\mathbb{R}^{qL}$ with $L = T_\mathrm{ini} + T_\mathrm{f}$ has to be consistent with the system behavior, which we write as $w\in\mathcal{B}_L$, with a slight abuse of notation.
In practice, measurements are often corrupted by noise, and therefore, 
only a noisy estimate $\hat{w}_\mathrm{ini}$ of $w_\mathrm{ini}$ is available.
To mitigate this, one can add $w_\mathrm{ini}$ to the problem as an optimization variable that has to be close (but not necessarily equal) to the measurement $\hat{w}_\mathrm{ini}$.
Moreover, the trajectory must be in some convex set of constraints $\mathcal{C}$.
Formalizing these requirements leads to the predictive control problem~\cite{dorfler2022bridging,markovsky2022data}
\begin{align} \label{eq:behavioral_control}
\begin{split}
    \min_{w\in \mathcal{B}_L \cap \mathcal{C}} & \quad  J(w_\mathrm{f}) + \lambda_\sigma \|\sigma\|_2^2, \\
    \mathrm{s.t.} & \quad  w_\mathrm{ini} = \hat{w}_\mathrm{ini} + \sigma.
    \end{split}
\end{align}
The controller is applied in receding horizon, i.e., we solve the problem, apply the first input $u_{t}$, then restart the process.

To highlight the connection between problem~\eqref{eq:behavioral_control} and projections, we substitute the constraint in the cost yielding
    \begin{align} \label{eq:behavioral_control_proj}
    \min_{w\in \mathcal{B}_L \cap \mathcal{C}} \left \|w - \begin{bmatrix}
    \hat{w}_\mathrm{ini} \\ w_\mathrm{ref}
    \end{bmatrix}\right \|^2_{W},
    \end{align}
where $w_\mathrm{ref}$ is the reference for $w_\mathrm{f}$ constructed from $u_\mathrm{ref}$ and $y_\mathrm{ref}$, and $W$ is a block diagonal symmetric positive definite weighting matrix with blocks built from $\lambda_\sigma I$, $Q$, and $R$.
Problem~\eqref{eq:behavioral_control_proj} is equivalent to the (weighted) projection of $\begin{bmatrix}\hat{w}_\mathrm{ini}^\top & w_\mathrm{ref}^\top\end{bmatrix}^\top$ onto the set $\mathcal{B}_L\cap\mathcal{C}$.
In case there are no constraints, i.e., $\mathcal{C} = \mathbb{R}^{qL}$, the solution is simply
$$w^\star_\mathcal{B} = P_B^W W\begin{bmatrix}\hat{w}_\mathrm{ini}^\top & w_\mathrm{ref}^\top\end{bmatrix}^\top,$$ where $B$ is a basis for the subspace $\mathcal{B}_L$, and
$P_B^W := B(B^\top W B)^{-1} B^\top$ is the projector onto $\mathcal{B}_L$ weighted by $W$.

\subsection{Data-driven predictive control} \label{sec:DeePC_sp}
The control formulation~\eqref{eq:behavioral_control} requires the exact knowledge of the system behavior $\mathcal{B}_L$, which is often not available.
In the remainder of the paper, we assume that the behavior is unknown, but $D \gg qL$ noisy trajectories $w^\mathrm{d}$ from $\mathcal{B}_L$ are available.
We arrange the measurements in a data matrix as
\begin{align} \label{eq:noisy_data}
     H = \begin{bmatrix}
         w^\mathrm{d}_{1,L} & w^\mathrm{d}_{2,L+1} & \dots & w^\mathrm{d}_{D,L+D-1}
     \end{bmatrix}.
\end{align}
The data matrix $H$ is a concatenation of two Hankel matrices, one for the inputs, and one for the outputs.
With a slight abuse of notation, we permute the rows of $H$ in the subsequent control problems, so that the first $qT_\mathrm{ini}$ rows correspond to $w_\mathrm{ini}$, and the rest corresponds to $w_\mathrm{f}$.

Given the noisy data $H$, the data-driven control problem, termed \gls{DeePC}~\cite{coulson2019data}, is formulated as
\begin{align} \label{eq:DeePC}
    \begin{split}
    \min_{w \in \mathcal{C},g\in\mathbb{R}^D} & \quad J(w_\mathrm{f}) + \lambda_\sigma \|\sigma\|_2^2 + \lambda_g\|g\|_2^2 \\
    \mathrm{s.t.} & \quad Hg = \begin{bmatrix}
        \hat{w}_\mathrm{ini} + \sigma \\ 
        w_\mathrm{f}
    \end{bmatrix}.
    \end{split}
\end{align}
If the measurements are noise free and the data is persistently exciting, then $\mathrm{col}(H) = \mathcal{B}_L$ by the fundamental lemma~\cite{willems2005note}.
Thus, the constraint $Hg=w$ enforces $w\in\mathcal{B}_L$.
However, $H$ becomes full rank in general when the data is noisy, rendering the constraint $Hg=w$ meaningless.
To address this issue, the regularization term $\|g\|_2^2$ with weighting $\lambda_g$ is often added to the cost.
Another commonly used regularizer, called the projected 2-norm regularizer~\cite{markovsky2022data,dorfler2022bridging}, is defined as $\|(I-\Pi)g\|_2^2$ with $\Pi:=
([Z_\mathrm{p}^\top~ U_\mathrm{f}^\top]^\top)^\dagger[Z_\mathrm{p}^\top~ U_\mathrm{f}^\top]^\top$, where $Z_\mathrm{p}$ and $U_\mathrm{f}$ denote the rows of $H$ corresponding to $w_\mathrm{ini}$ and the inputs in $w_{f}$, respectively.

In the absence of constraints, the solution to~\eqref{eq:DeePC} can be expressed in closed form as a ``softened" projection onto $\mathrm{col}(H)$, as shown below.
\begin{lemma} \label{lemma:DeePC}
    If $\mathcal{C} = \mathbb{R}^{qL}$ and $\lambda_g>0$, the solution to~\eqref{eq:DeePC} is
    \begin{align*}
    g^\star & = (H^\top W H + \lambda_g I)^{-1}H^\top W \begin{bmatrix}
            \hat{w}_\mathrm{ini} \\ w_\mathrm{ref}
        \end{bmatrix}; \quad w^\star_\mathrm{DeePC} = Hg^\star.
    \end{align*}
\end{lemma}
\begin{proof}
Substituting the constraint in the cost in~\eqref{eq:DeePC} gives
$$
\min_{g\in\mathbb{R}^{D}}\left\|Hg - \begin{bmatrix}
    \hat{w}_\mathrm{ini} \\ w_\mathrm{ref}
\end{bmatrix}\right\|_W^2 + \lambda_g\|g\|_2^2.
$$
The statement follows using the fact that this is a regularized least-squares problem~\cite[Section 15.1]{boyd2018introduction}.    
\end{proof}

Lemma~\ref{lemma:DeePC} reveals the connection between problem~\eqref{eq:behavioral_control} and~\eqref{eq:DeePC}.
Namely, the true projector $P_B^W$ in~\eqref{eq:behavioral_control_proj} is approximated by $H(H^\top W H + \lambda_g I)^{-1} H^\top$ in \gls{DeePC}, which can be interpreted as a weighted soft projector onto the column space of $H$.
Since noise in the data makes $H$ full rank in general, the true projection to $\mathrm{col}(H)$ becomes identity.
To resolve this issue, the projector in the data-driven setting is ``softened" by the regularizer $\lambda_g$.

Even though \gls{DeePC} has been observed to achieve great performance in various case studies~\cite{dorfler2022bridging,markovsky2021behavioral}, the formulation has the following limitations.
1) Regularization on the $g$ vector is difficult to interpret, as it is not connected to the  behavior directly.
The interpretation through weighted soft projections above only holds in the absence of constraints, as \gls{DeePC} attains an explicit solution in this case. 
2) The (data-driven or non-parametric) system model only appears implicitly through the regularizer in~\eqref{eq:DeePC}. 
As a consequence, it is not straightforward how new data can be incorporated in the formulation for online adaptation.

Motivated by these observations, we further analyze the soft projector next and show that it is a good approximation of the true projection onto the behavior.
We also show that formulating data-driven predictive control problems using soft projectors explicitly resolves the above mentioned limitations of \gls{DeePC}, and gives rise to novel control schemes.

\section{Soft projections} \label{sec:soft_proj}
Let us analyze the soft projector with $\delta>0$, defined as
\begin{align} \label{eq:soft_proj}
    \tilde{P}_H^W := H(H^\top W H + \delta I)^{-1} H^\top,
\end{align}
with positive definite weight $W$.
The unweighted soft projector with $W=I$ was introduced in~\cite{wax2021detection,wax2022vector,wax2023robust}, and it is closely related to \textit{diagonal loading}, which is a widely used technique in the signal processing literature~\cite{van2002optimum}.
We denote the unweighted soft projector by $\tilde{P}_H := \tilde{P}_H^I$ throughout the paper.
We derive the results for $\tilde{P}_H^W$ for completeness, but often focus on $\tilde{P}_H$ to build intuition.

Computing $\tilde{P}_H^W$ as in~\eqref{eq:soft_proj} involves inverting a $D$-by-$D$ matrix, which can be computationally restrictive. 
Using the matrix inversion lemma, $\tilde{P}_H^W$ can also be expressed as~\cite{wax2023robust}
\begin{align} \label{eq:soft_proj_covariance}
    \tilde{P}_H^W = W^{-1} - \left(\frac{1}{\delta}WHH^\top W + W\right)^{-1}.
\end{align}
This formulation only requires inverting an $qL$-by-$qL$ matrix, and thus, the computation does not scale with the data size.

As the constant $\delta\to 0$, the soft projector $\tilde{P}_H$ reduces to the orthogonal projector onto $\mathrm{col}(H)$, i.e., $\tilde{P}_H \to HH^\dagger = P_H$, where $\dagger$ denotes the Moore-Penrose inverse.
In the noise-free case we have that $\mathrm{col}(H) = \mathrm{col}(B)$, and thus, the soft projector recovers the true projection $P_B$ as $\delta\to0$.
Recall that the eigenvalues of an orthogonal projector are either 1 or 0.
For $\delta>0$, however, the soft projector $\tilde{P}_H$ scales the singular values of $H$, denoted by $\sigma_i$.
The eigenvalues of $\tilde{P}_H$ are $0\leq\dfrac{\sigma_i^2}{\sigma_i^2 + \delta}<1$, while the eigenvectors are the same as those of $HH^\top$.
As $\delta$ increases, the larger $\sigma_i$ start to dominate in $\tilde{P}_H$, and the smaller $\sigma_i$ are suppressed.
Thus, the soft projector ``denoises" the data matrix $H$.
Since all eigenvalues of $\tilde{P}_H$ are less than one, the soft projections are biased towards zero compared to the true projection $P_B$.

\begin{rmk} \label{rmk:DeePC_explicit_sol}
Introducing the soft projector $\tilde{P}_H$ was motivated by the solution to the \gls{DeePC} problem with 2-norm regularization in Section~\ref{sec:DeePC_sp}.
In fact, the solution to \gls{DeePC} in Lemma~\ref{lemma:DeePC} can be expressed as $w^\star_\mathrm{DeePC} = \tilde{P}^W_H W\begin{bmatrix}
    \hat{w}_\mathrm{ini}^\top & w_\mathrm{ref}^\top
\end{bmatrix}^\top$ with $\delta = \lambda_g$.
In case the projected 2-norm regularizer is used in~\eqref{eq:DeePC}, the soft projector becomes
$$
\hat{P}_H^{W} = H\left(H^\top W H + \delta(I-P_{[Z_p^\top~U_f^\top]})\right)^{-1} H^\top,
$$
where the inverse exists, since {$\mathrm{col}([Z_p^\top~U_f^\top]^\top) \subseteq \mathrm{col}(H^\top)$}.
The projector $\hat{P}_H^{W}$ does not shrink vectors in the subspace $\mathrm{col}([Z_p^\top~U_f^\top])$, only in its orthogonal complement.
In the following, we focus on the soft projector $\tilde{P}_H^{W}$, but similar arguments hold for $\hat{P}_H^{W}$. \hfill \QED
\end{rmk}

\subsection{Approximation error bounds}
We now analyze the estimation error of the weighted soft projector.
Motivated by the gap metric, we quantify the approximation error of the weighted soft projection by providing an upper bound on $\|\tilde{P}_H^W - P_B^W\|_2$ below. For the analysis, let us decompose the data matrix as $H = BS + E$, where $B\in\mathbb{R}^{qL\times d}$, $BB^\top = I$ is an orthonormal basis for the subspace $\mathcal{B}_L$. 
We interpret the term $BS$ as the nominal part of the data that is in the behavior, and $E$ is noise.
Norm bounds on $E$ can be derived from energy bounds~\cite{van2023informativity} or confidence ellipsoid level sets~\cite{sasfi2025gaussian}.
\begin{thm} \label{thm:bound}
Assume that $\delta>0$ and $W,S$ have full rank. Then, $\|\tilde{P}_H^W - P_B^W\|_2 \leq \gamma(\delta)$, with
    \begin{align*}
    \small
    \begin{split}
        \gamma(\delta):=&\kappa(W)\frac{2\|S\|_2\|E\|_2 + \|E\|_2^2}{\sigma_\mathrm{min}^2(W^{1/2}H) + \delta} + \frac{\delta \|W^{-1}\|_2}{\sigma_\mathrm{min}(B^\top W BSS^\top) + \delta},
    \end{split}
    \end{align*}
    where $\sigma_\mathrm{min}(\cdot)$ ($\sigma_\mathrm{max}(\cdot)$) denotes the smallest (largest) singular value of a matrix, and $\kappa(W) = \sigma_\mathrm{max}(W)/\sigma_\mathrm{min}(W)$ is the condition number of $W$.
\end{thm}

\if\arxiv0
The proof can be found in the extended version~\cite{sasfi2026soft}.
\fi
\if\arxiv1
The proof can be found in the Appendix. 
\fi
The first term in $\gamma(\delta)$ quantifies the effect of the noise $E$. The second term is the bias that is introduced through $\delta$.
The term $\sigma_\mathrm{min}(B^\top WBSS^\top)$ is related to the minimal signal content of the data, and it reduces the bias term in $\gamma(\delta)$.
Since the projectors are weighted by the matrix $W$, the approximation error also depends on the conditioning of $W$.

Clearly, soft projectors are consistent in the sense that $\|\tilde{P}_H^W - P_B^W\|_2 \to 0$ as $\|E\|_2 \to 0$ and $\delta\to0$, and the error bound changes smoothly with $\delta$.
Importantly, the approximation error does not depend on the true dimension of $\mathcal{B}_L$.
This is in contrast with the approach to data-driven control in the behavioral setting that explicitly estimates $\mathcal{B}_L$ as a low-dimensional subspace from data~\cite{sasfi2025great,bharadwaj2025robust,jin2026sensitivitysubspacepredictorbehavioral}.
If the estimated subspace and $\mathcal{B}_L$ and are of different dimensions, the gap metric is 1, meaning that the estimate can be arbitrarily bad.
Hence, this approach is sensitive to order selection, which is a major limitation.

Theorem~\ref{thm:bound} enables us to provide a bound on the error introduced in \gls{DeePC}~\eqref{eq:DeePC} in the unconstrained case due to using noisy data to approximate problem~\eqref{eq:behavioral_control_proj}.
\begin{cor}
    Let $w^\star_\mathrm{DeePC}$ and $w^\star_\mathcal{B}$ be the solutions to~\eqref{eq:DeePC} and \eqref{eq:behavioral_control_proj}, respectively, with $\mathcal{C}=\mathbb{R}^{qL}$.
    Then,
    $$
    \|w^\star_\mathrm{DeePC} - w^\star_\mathcal{B}\| \leq \gamma(\delta)\left\|W \begin{bmatrix}
        \hat{w}_\mathrm{ini} \\ w_\mathrm{ref}
    \end{bmatrix} \right\|.
    $$
\end{cor}

\section{Data-driven control using soft projections} \label{sec:novel_methods}
So far, we interpreted the solution of~\eqref{eq:DeePC} through a soft projection in the unconstrained case, i.e., when $\mathcal{C} = \mathbb{R}^{qL}$.
In this section, we propose a novel approach that allows incorporating the constraint $w\in\mathcal{C}$, while approximately enforces the behavioral constraint $w \in \mathcal{B}_L$ in problem~\eqref{eq:behavioral_control_proj}.
If a basis $B$ for the behavior is known, the constraint $w \in \mathcal{B}_L$ can be written as $\|(I - P_B) w\| = 0$.
For the soft projector, however, this only holds if $w=0$ in general.
Therefore, we formulate the data-driven control problem by lifting the approximated constraint into the cost, leading to the following convex problem
\begin{align} \label{eq:indicator_approx_squared}
     \min_{w\in\mathcal{C}} \left \|w - \begin{bmatrix}
    \hat{w}_\mathrm{ini} \\ w_\mathrm{ref}
    \end{bmatrix}\right \|^2_{W} + \alpha \|(I - \tilde{P}_H)w\|^2. 
\end{align}
In the above formulation, $\alpha$ is a tuning parameter that controls how strong the approximation of the constraint $w\in\mathcal{B}_L$ is enforced.
Furthermore, $\delta$ in $\tilde{P}_H$ controls how much the data is ``denoised" by softening the projection.
Note that if the data is noise-free and persistently exciting, $\tilde{P}_H\to P_B$ as $\delta\to0$, and thus,~\eqref{eq:indicator_approx_squared} reduces to~\eqref{eq:behavioral_control_proj} as $\alpha\to\infty$. 

To connect our approach to existing frameworks, we now consider an alternative formulation.
For orthogonal projectors, one can equivalently express the constraint $w\in\mathcal{B}_L$ as $w^\top(I-P_B)w = 0$, since $P_B^2 = P_B$.
However, this is not true for the soft projector $\tilde{P}_H$.
Thus, approximating the constraint $w\in\mathcal{B}_L$ by penalizing $w^\top(I-\tilde{P}_H)w$ gives rise to the alternative control problem
\begin{align} \label{eq:indicator_approx}
     \min_{w\in\mathcal{C}} \left \|w - \begin{bmatrix}
    \hat{w}_\mathrm{ini} \\ w_\mathrm{ref}
    \end{bmatrix}\right \|^2_{W} + \frac{\hat{\alpha}}{\delta} w^\top(I - \tilde{P}_H)w. 
\end{align}
We scaled the tuning parameter $\hat{\alpha}$ by $\delta$ to highlight the connection between~\eqref{eq:indicator_approx} and \gls{DeePC}~\eqref{eq:DeePC} as formalized below.
\begin{thm} \label{thm:generalizing_DeePC}
   Let $w^\star_\mathrm{DeePC}$ and $w^\star_\eqref{eq:indicator_approx}$ be the solutions to~\eqref{eq:DeePC} and \eqref{eq:indicator_approx}, respectively, with $\mathcal{C}=\mathbb{R}^{qL}$.
   Assume $H$ is full rank and $\lambda_g,\delta>0$.
   If $\hat{\alpha} = \lambda_g$, then $w^\star_\eqref{eq:indicator_approx} \to w^\star_\mathrm{DeePC}$ as $\delta \to 0$. 
\end{thm}

\if\arxiv0
The proof can be found in the extended version~\cite{sasfi2026soft}.
\fi
\if\arxiv1
\begin{proof}
    We first reformulate the solution $w^\star_\mathrm{DeePC}$ from Lemma~\ref{lemma:DeePC}.
    Using the identity $(AB + I)^{-1}A = A(BA +I)^{-1}$ with $A = H^\top$ and $B=\lambda_g^{-1}WH$ leads to
    \begin{align*}
        &H(H^\top W H +  \lambda_g I)^{-1} H^\top  = \lambda_g^{-1} H(AB + I)^{-1}A \\
        &= \lambda_g^{-1} H A(BA+I)^{-1} = \left((BA +I)\cdot\lambda_g(HA)^{-1}\right)^{-1} \\
        & = \left(W + \lambda_g(HH^\top)^{-1}\right)^{-1}.
    \end{align*}
    Thus, $w^\star_\mathrm{DeePC} = \left(W + \lambda_g(HH^\top)^{-1}\right)^{-1}W\begin{bmatrix}
        \hat{w}_\mathrm{ini}^\top & w_\mathrm{ref}^\top
    \end{bmatrix}^\top.$
    Furthermore,~\eqref{eq:indicator_approx} without constraints is a regularized least-squares problem, whose solution is
    \begin{align*}
    w^\star_\eqref{eq:indicator_approx} &= \left(W + \hat{\alpha}/\delta \cdot(I-\tilde{P}_H)\right)^{-1}W\begin{bmatrix}
        \hat{w}_\mathrm{ini}^\top & w_\mathrm{ref}^\top
    \end{bmatrix}^\top.
    \end{align*}
    Inserting~\eqref{eq:soft_proj_covariance} into $w^\star_\eqref{eq:indicator_approx}$ proves the claim.
    Note that the inverses exist for any $\delta>0$, since $\hat{\alpha}>0$, and $H$ is full row rank.
\end{proof}
\fi
Theorem~\ref{thm:generalizing_DeePC} shows that~\eqref{eq:indicator_approx} is a generalization of the \gls{DeePC} formulation.
Even though we showed in Section~\ref{sec:DeePC_sp} that \gls{DeePC} performs a weighted soft projection in the unconstrained case, this interpretation is lost when constraints $\mathcal{C}$ are considered.
On the other hand, the interpretation of \eqref{eq:indicator_approx_squared} and \eqref{eq:indicator_approx} through approximate data-driven projection onto the behavior remain valid even in the constrained case.
Next, we analyze the difference between the two proposed schemes.

\subsection{Comparison of the proposed control formulations}
Recall that the solution to~\eqref{eq:behavioral_control_proj} in the unconstrained case is $w^\star_\mathcal{B} = P_B^W W\begin{bmatrix}
    \hat{w}_\mathrm{ini}^\top & w_\mathrm{ref}^\top
\end{bmatrix}^\top$. 
The matrix $P_B^W W$ mapping the desired trajectory to the solution is low rank, reflecting that any feasible trajectory lies in the low-dimensional subspace $\mathcal{B}_L$.
When the solution is estimated from noisy data through soft projections, the low rank structure is destroyed.
Yet, the separation of the eigenvalues is desired to approximate the solution to~\eqref{eq:behavioral_control_proj} well.
As we established in Section~\ref{sec:soft_proj}, the eigenvalues of the soft projector can be filtered by tuning the parameter $\delta$.
Intuitively speaking, the term $(I - \tilde{P}_H)^2$ in~\eqref{eq:indicator_approx_squared} acts as a higher order filter that can separate the signal from noise in $H$ better than $(I - \tilde{P}_H)$ in~\eqref{eq:indicator_approx}.

We illustrate this difference by analyzing the simple case when $\mathcal{C}=\mathbb{R}^{qL}$ and $W=I$.
The proposed formulations~\eqref{eq:indicator_approx_squared} and \eqref{eq:indicator_approx} boil down to regularized least-squares problems, whose solutions are
$M_\eqref{eq:indicator_approx_squared}\begin{bmatrix}
    \hat{w}_\mathrm{ini}^\top & w_\mathrm{ref}^\top
\end{bmatrix}^\top$ and $M_\eqref{eq:indicator_approx}\begin{bmatrix}
    \hat{w}_\mathrm{ini}^\top & w_\mathrm{ref}^\top
\end{bmatrix}^\top$, respectively, where
\begin{align*}
    M_\eqref{eq:indicator_approx_squared} & = \left(I + \alpha\delta^2(HH^\top +\delta I )^{-2}\right)^{-1}, \\
    M_\eqref{eq:indicator_approx} & = \left(I + \hat{\alpha}(HH^\top+\delta I)^{-1} \right)^{-1}.
\end{align*}
The eigenvalues of the maps are 
\begin{align*}
    \lambda_i(M_\eqref{eq:indicator_approx_squared}) = \frac{(\sigma_i^2 + \delta)^2}{(\sigma_i^2 + \delta)^2 + \alpha\delta^2}; \quad \lambda_i(M_\eqref{eq:indicator_approx}) = \frac{\sigma_i^2 + \delta}{\sigma_i^2 + \delta + \hat{\alpha}},
\end{align*}
where $\sigma_i$ are the singular values of $H$.
Figure~\ref{fig:ev_plot} displays the eigenvalues for $\sigma_i=\{1000,100,50,30,20\}$ and $\alpha = \hat{\alpha}/\delta = 10^6$ as a function of $\delta$. 
Interpreting $(I - \tilde{P}_H)^2$ as a second order filter, we see that the eigenvalues of $M_\eqref{eq:indicator_approx_squared}$ are separated more clearly. 
This provides the intuition that formulation~\eqref{eq:indicator_approx_squared} is more robust to noise than formulation~\eqref{eq:indicator_approx}. 
In Section~\ref{sec:case_study}, we present a case study empirically confirming that this is indeed the case.
A comprehensive analysis of the differences between~\eqref{eq:indicator_approx_squared} and~\eqref{eq:indicator_approx} is subject of future work.

\begin{figure}
    \centering
    \includegraphics[width=0.8\linewidth,trim= 36 8 36 20,clip]{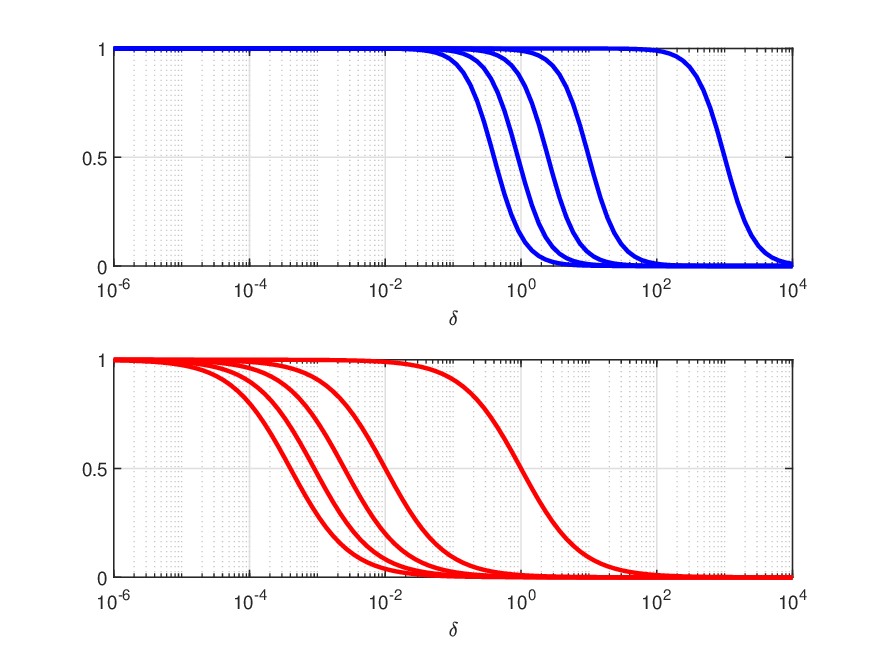}
    \caption{Eigenvalues of the closed-form solution mapping for~\eqref{eq:indicator_approx_squared} (top) and \eqref{eq:indicator_approx} (bottom) in the unconstrained case with $W=I$ and $\alpha = \hat{\alpha}/\delta = 10^6$. The singular values of the data matrix $H$ are $1000,100,50,30,$ and $20$. \vspace{-5mm}}
    \label{fig:ev_plot}
\end{figure}

\section{Online data-driven control} \label{sec:online}
A major benefit of the proposed formulations~\eqref{eq:indicator_approx_squared}--\eqref{eq:indicator_approx} is that $\tilde{P}_H$ serves as a data-driven model of the system, which can easily be updated online.
The continuous adaptation of these methods is essential in many application domains, as real-world dynamical systems often change over time.
Furthermore, adapting the proposed linear schemes is particularly important because nonlinear systems can be effectively over-approximated as linear time-varying systems~\cite{berberich2024overview}.

Consider the scenario where the data matrix in~\eqref{eq:noisy_data} includes all trajectories observed up to the current time $t$, denoted as
$$
H_t = \begin{bmatrix}
    w^\mathrm{d}_{1,L} &w^\mathrm{d}_{2,L+1} & \dots & w^\mathrm{d}_{t-L,t}
\end{bmatrix}.
$$
While this time-varying matrix can be directly incorporated into the data-driven control problems~\eqref{eq:DeePC},~\eqref{eq:indicator_approx_squared}, or~\eqref{eq:indicator_approx}, such a naive implementation presents two primary challenges.
First, it is computationally inefficient because the problem size in~\eqref{eq:DeePC} scales with the number of columns in $H_t$, or the soft projector $\tilde{P}_{H_t}^W$ in formulations~\eqref{eq:indicator_approx_squared} and~\eqref{eq:indicator_approx} would need to be recomputed at every time step.
Second, the appropriate selection of parameters $\lambda_g$ or $\delta$ depends directly on the current data matrix $H_t$ to achieve a favorable bias-variance trade-off (cf. Theorem~\ref{thm:bound}).
Since $H_t$ evolves as new data is collected, it is necessary to re-tune the parameter to maintain good performance.
To address these challenges, we develop efficient update formulas for both the soft projector $\tilde{P}^W_{H_t}$ and the tuning parameter $\delta_t$. This enables direct tracking of the unconstrained \gls{DeePC} solution and allows the novel formulations from Section~\ref{sec:novel_methods} to be adapted online.
Inspired by~\cite{wax2023robust}, we present an efficient rank-one update for $\tilde{P}^W_{H_t}$.
\begin{prop} \label{prop:update}
    Given the most recent data trajectory $w^\mathrm{d}_{t-L,t}$, the weighted soft projector at time $t$ is
    \begin{align} \label{eq:req_formula_weighted}
        &\qquad \quad \tilde{P}_{H_t}^W = \tilde{P}_{H_{t-1}}^W + \frac{1}{\gamma_{t} \delta_{t-1}} \tilde{w}_t \tilde{w}_t^\top, \text{ where} \\
        \tilde{w}_t &= (I-\tilde{P}_{H_{t-1}}^W W) w^\mathrm{d}_{t-L,t}; \quad  \gamma_{t} = 1 + \frac{1}{\delta_{t-1}} \tilde{w}_t^\top W w^\mathrm{d}_{t-L,t}; \nonumber \\
        \delta_t & = \epsilon \cdot \mathrm{tr}(H_t H_t^\top) = \delta_{t-1} + \epsilon \|w^\mathrm{d}_{t-L,t}\|^2. \nonumber
    \end{align}
\end{prop}
Proposition~\ref{prop:update} can be derived following the same steps as in~\cite{wax2023robust}, and for the unweighted case ($W=I$), the update rule~\eqref{eq:req_formula_weighted} reduces to the formula proposed therein.
At time $t-1$, the soft projector $\tilde{P}_{H_{t-1}}^W$ serves as an estimate of the true projector $P_{\mathcal{B}_L}^W$.
The vector $\tilde{w}_t$ captures the estimated projection error of the newly observed trajectory $w^\mathrm{d}_{t-L,t}$.
To incorporate this new information, the soft projector is corrected by adding the rank-one term $\tilde{w}_t\tilde{w}_t^\top$, appropriately normalized by the scalar $1 / (\gamma_t \delta_{t-1})$. 
Furthermore, as argued in~\cite{wax2021detection,wax2023robust}, setting $\delta_t = \epsilon \cdot \mathrm{tr}(H_t H_t^\top)$ for a small, possibly time-varying scalar $\epsilon$ ensures that $\tilde{P}_{H_t}^W$ remains a good approximation of the true projector.
Finally, for time-varying systems, the update rule can be readily modified to include a forgetting factor that discounts older data~\cite{wax2023robust}.

The recursive update formulas in Proposition~\ref{prop:update} enable the efficient online implementation of the proposed methods. 
Given an initial time $t_0 \geq d$, we compute $\tilde{P}_{H_{t_0}}^W$ using~\eqref{eq:soft_proj} and set $\delta_0 = \epsilon \cdot \mathrm{tr}(H_{t_0}H_{t_0}^\top)$. 
Then, for each $t > t_0$, we
\begin{enumerate}[leftmargin=*]
    \item set $\hat{w}_\mathrm{ini} = w_{t-T_\mathrm{ini}-1,t-1}$ for the controller
 \item express the unconstrained \gls{DeePC} solution $w^\star_\mathrm{DeePC} = \tilde{P}^W_{H_{t-1}} W\begin{bmatrix}
    \hat{w}_\mathrm{ini}^\top & w_\mathrm{ref}^\top
\end{bmatrix}^\top$, or solve ~\eqref{eq:indicator_approx_squared} or \eqref{eq:indicator_approx} with $\tilde{P}_{H_{t-1}}^W$
    \item extract the current input $u_{t}$ from the solution, apply it to the system, and measure the corresponding output $y_{t}$
    \item update the most recent data trajectory $w^\mathrm{d}_{t-L,t} = \begin{bmatrix}
        u_{t-L}^\top & \dots & u_{t-1}^\top & u_{t}^\top & y_{t-L}^\top & \dots & y_{t-1}^\top & y_{t}^\top
    \end{bmatrix}^\top$
    \item \label{step_last} update $\tilde{P}_{H_t}^W$ and $\delta_t$ using the formulas in Proposition~\ref{prop:update}.
\end{enumerate}

The adaptation scheme outlined above provides a computationally efficient alternative to the naive implementation by relying on rank-one updates.
Specifically, the adaptation of new data in step~\ref{step_last} is highly efficient, as its most expensive operation (calculating the vector $\tilde{w}_t$ in Proposition~\ref{prop:update}) requires only $\mathcal{O}((qL)^2)$ operations.
Furthermore, the scheme systematically manages the time-varying tuning of the regularizer $\delta_t$.
Nevertheless, updating the soft projector with closed-loop data requires care, as it can lead to a loss of persistence of excitation and degrade the data-driven model.
Such issues can be resolved by adding a small probing signal to $u_t$ or updating $\tilde{P}^W_{H_t}$ only when the residual 
$\tilde{w}_t$ is sufficiently informative, see~\cite{astrom2008adaptive} for a detailed discussion.

\section{Case study} \label{sec:case_study}
We consider the double spring-mass-damper system from~\cite{kerz2023data} consisting of two rotating discs.
The MATLAB code reproducing the results is available online at {\url{https://gitlab.ethz.ch/asasfi/SP_for_robust_DDPC
}}.
The states are the angles and angular velocities of the two disks, the outputs are the two angles, and the input is the torque on the first disk. 
Zero-mean white noise torques on both disks act as disturbances.
The objective is to track the reference at $\begin{bmatrix} 0.8 & 0.8\end{bmatrix}^\top$, starting from zero initial state.
We define the constraint set $\mathcal{C}$ by $\begin{bmatrix}-1&-1\end{bmatrix}^\top \leq y_k \leq \begin{bmatrix}1&1\end{bmatrix}^\top$ and $-2\leq u_k \leq 2$ for all $k=t,\dots,t+T_\mathrm{f}-1$.
To gather data, we simulate the system for $1000$ steps starting from zero initial state and applying zero mean white noise as input.
We use the following parameters $T_\mathrm{ini} = 2$, $T_\mathrm{f} = 12$, $Q = I,~R = 0.01\cdot I, \lambda_\sigma = 10^6$, and $\alpha=10^6$.

We compare the open-loop predictions from \gls{DeePC}~\eqref{eq:DeePC} and~\eqref{eq:indicator_approx_squared}.
Problem~\eqref{eq:indicator_approx} performed similarly to \gls{DeePC}, thus, the results are omitted.
The parameters $\lambda_g$ and $\delta$ are tuned through validation from the intervals $[10,10^7]$ and $[10^{-3},10^3]$, respectively.
The open-loop prediction error and realized cost during validation are plotted in Figure~\ref{fig:validation} as the function of $\lambda_g$ (or $\delta$ for~\eqref{eq:indicator_approx_squared}).
We choose $\lambda_g = 54.29$ and $\delta = 0.091$, which minimize the open-loop cost on the validation set, and test both methods with these values.
Validation and testing use 100 noise realizations each, and the experiment is repeated for signal-to-noise ratios (SNR) 10, 5, and 3.
Table~\ref{tab:snr10} shows the mean realized open-loop cost and the prediction error statistics.
Figure~\ref{fig:validation} shows that $\lambda_g$ trades off prediction error against realized cost for \gls{DeePC}, whereas for~\eqref{eq:indicator_approx_squared} both are minimized on roughly the same interval $[10^{-2},1]$.
Both Figure~\ref{fig:validation} and the results in Table~\ref{tab:snr10} show that the proposed method~\eqref{eq:indicator_approx_squared} achieves superior performance compared to \gls{DeePC}.
As the signal-to-noise ratio shrinks, the difference between the mean realized costs grows, suggesting that~\eqref{eq:indicator_approx_squared} overperforms \gls{DeePC} due to its increased robustness.

\begin{table}
\begin{tabular}{r|c|c|c|c}
     SNR & method & mean cost & mean pred. err. & pred. err. variance  \\
     \hline \hline
     \multirow{2}{2em}{10} & \eqref{eq:DeePC} & 5.0535 & 0.61382 & 0.04813 \\
& \eqref{eq:indicator_approx_squared} & 4.6658 & 0.24934 & 0.021772 \\
\hline
\multirow{2}{1.5em}{5} & \eqref{eq:DeePC} & 6.843 & 0.70232 & 0.14774 \\
& \eqref{eq:indicator_approx_squared} & 5.4872 & 0.5146 & 0.093555 \\
\hline
\multirow{2}{1.5em}{3} & \eqref{eq:DeePC} & 9.4343 & 1.3884 & 0.47615 \\
& \eqref{eq:indicator_approx_squared}&  6.8728 & 0.94667 & 0.30057
\end{tabular}
\caption{Realized open-loop cost and prediction error statistics for \gls{DeePC}~\eqref{eq:DeePC} and the proposed method~\eqref{eq:indicator_approx_squared}. \vspace{-5mm}}
    \label{tab:snr10}
\end{table}

\begin{figure}
    \centering
    \includegraphics[width=1\linewidth,trim= 28 0 30 2,clip]{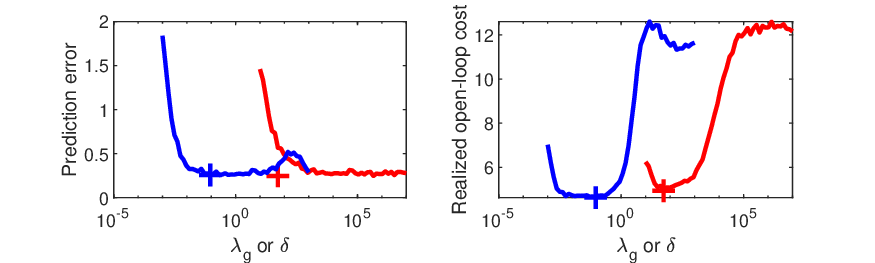}
    \caption{Prediction error and realized cost for \gls{DeePC} (red) and the proposed method (blue) during validation under SNR = 10. The cross denotes the chosen values. \vspace{-5mm}}
    \label{fig:validation}
\end{figure}

\section{Conclusion} \label{sec:conclusion}
We introduced the use of soft projectors to approximate the true behavioral projector directly from noisy data, providing an alternative to regularization in direct data-driven predictive control.
We established a bound on the approximation error that reveals a bias-variance trade-off and is independent of the system order. 
Based on these insights, we developed generalizations of regularized \gls{DeePC} that showed superior performance in a case study.
Finally, we derived efficient update formulas that allow the proposed schemes to be adapted online with streaming data.
Future work includes investigating further approaches to estimate the behavioral projection and their interpretation in the stochastic setting. 

\if\arxiv1
\section*{Appendix}
\appendices
\begin{proof}[Theorem~\ref{thm:bound}]
    The triangle inequality gives $\|\tilde{P}^W_H - P^W_B\|_2 \leq \|\tilde{P}^W_H - \tilde{P}^W_{BS}\|_2 + \|\tilde{P}_{BS}^W - P_B^W\|_2$. Let us define $W_B := B^\top W B$. As $B^\top B =I$, we have
    \begin{align*}
    \|\tilde{P}_{BS}^W & - P_B^W\|_2 \\
    & = \|B\left(S(S^\top W_B S+\delta I)^{-1}S^\top - W_B^{-1}\right)B^\top\|_2 \\ 
    & = \|\tilde{P}_S^{W_B} - W_B^{-1}\|_2 = \|(1/\delta \cdot W_BSS^\top W_B + W_B)^{-1}\|_2 \\
    & \leq \|W_B^{-1}\|_2\cdot \|(1/\delta \cdot W_BSS^\top + I)^{-1}\|_2\\
    & = \|W_B^{-1}\|_2 \cdot \sigma_\mathrm{min}^{-1}(1/\delta \cdot W_B SS^\top + I) \\
    & = \frac{\delta \|W_B^{-1}\|_2}{\sigma_\mathrm{min}(W_B SS^\top ) + \delta} 
    \leq \frac{\delta \|W^{-1}\|_2}{\sigma_\mathrm{min}(W_B SS^\top ) + \delta}.
    \end{align*}
    Let us define $C_1 :=1/\delta\cdot W^{1/2}HH^\top W^{1/2}+I$ and $C_2:= 1/\delta\cdot W^{1/2}BSS^\top B^\top W^{1/2}+I$. The first term in the triangular inequality is then
    \begin{align*}
    \|\tilde{P}_H^W -  \tilde{P}_{BS}^W\| &= \|W^{-1/2}(C_2^{-1} - C_1^{-1})W^{-1/2}\| \\
    & \leq \|W^{-1}\| \cdot \|C_1^{-1}(C_2 - C_1)C_2^{-1}\| \\
    & \leq \|W^{-1}\| \cdot \|C_1^{-1}\|\cdot \|C_1 - C_2\|\cdot \|C_2^{-1}\| \\
    & \leq \|W^{-1}\|\cdot\|W\|\frac{\|HH^\top - BSS^\top B^\top\|}{\sigma^2_\mathrm{min}(W^{1/2}H) + \delta} \\
    & \leq \kappa(W) \frac{2\|BS\|\cdot\|E\| + \|E\|^2}{\sigma_\mathrm{min}^2(W^{1/2}H) + \delta},
    \end{align*}
    where we used that $\|C_2^{-1}\| = 1$, $\|C_1^{-1}\| \leq \frac{\delta}{\sigma_\mathrm{min}^2(W^{1/2}H) + \delta}$, and $\|C_1-C_2\|\leq 1/\delta \|W\|\cdot\|HH^\top - BSS^\top B^\top\|$ in the third inequality.
\end{proof}

\fi

\bibliographystyle{ieeetr}        
\bibliography{Literature}
\end{document}